\documentclass[conference,10pt]{IEEEtran}
\usepackage{cite}

\ifCLASSINFOpdf
  \usepackage[pdftex]{graphicx}
\else
\fi
\usepackage{color}
\usepackage{caption}
\usepackage{subcaption}

\usepackage[cmex10]{amsmath}
\usepackage{amsthm}
\usepackage{amssymb}

\newtheorem{prop}{Proposition}

\usepackage{stfloats}

\usepackage{threeparttable}

\begin{document}
%
\title{Enabling Complexity-Performance Trade-Offs for\\Successive Cancellation Decoding of Polar Codes}
%
%
%

\author{\authorblockN{Alexios Balatsoukas-Stimming, Georgios Karakonstantis, and Andreas Burg}\\%
\authorblockA{Telecommunications Circuits Laboratory, EPFL, Lausanne, Switzerland.\\}%
Email: \{alexios.balatsoukas, georgios.karakonstantis, andreas.burg\}@epfl.ch}%
\maketitle

\begin{abstract}
Polar codes are one of the most recent advancements in coding theory and they have attracted significant interest. While they are provably capacity achieving over various channels, they have seen limited practical applications. Unfortunately, the successive nature of successive cancellation based decoders hinders fine-grained adaptation of the decoding complexity to design constraints and operating conditions. In this paper, we propose a systematic method for enabling complexity-performance trade-offs by constructing polar codes based on an optimization problem which minimizes the complexity under a suitably defined mutual information based performance constraint. Moreover, a low-complexity greedy algorithm is proposed in order to solve the optimization problem efficiently for very large code lengths.
\end{abstract}


%
\IEEEpeerreviewmaketitle

\section{Introduction}
%
%
%
%

Channel polarization gives rise to elegant and provably optimal channel codes, called \emph{polar codes}~\cite{Arikan2009}, whose decoding complexity under successive cancellation (SC) decoding is $O(N\log N)$, where $N$ is the blocklength of the code. In implementations of channel decoders it is often desirable to trade performance for complexity in order to meet system requirements at minimal cost. These trade-off decisions can be offline, i.e., taken during the design phase of the system, or online, i.e., taken by the system during its operation. A trivial way to vary the decoding complexity is to alter the blocklength of the code. This is usually an offline decision, although some recent communications standards (e.g., IEEE 802.11n \cite{IEEE802.11n}) require the support of codes of various lengths. Unfortunately, using this method, polar codes do not offer a fine trade-off granularity because their blocklength is constrained to powers of two when using the simple $2 \times 2$ polarizing matrix introduced by Ar{\i}kan. Thus, any complexity reduction may lead to a large loss in performance. Moreover, due to its successive nature, the SC decoding algorithm is not amenable to online complexity tuning. Other codes which are used in modern systems, such as LDPC~\cite{Gallager1963} and turbo codes \cite{Berrou1996}, which are usually decoded using iterative decoding algorithms, can be tuned online by varying the number of performed iterations, according to e.g., the channel conditions. Fortunately, a simple observation allows us to trade decoding complexity for performance for SC decoding in small steps, both offline and online, without the need to change the blocklength by altering the set of channels which are used to transmit information.

\subsubsection*{Contribution} In this work, we populate the complexity-performance trade-off curve for SC decoding by formulating the frozen channel selection step of polar code construction as an optimization problem. This is achieved by reformulating the original problem of polar code construction with the objective to minimize the complexity while respecting quality constraints that represent the various dynamically changing operating conditions, or the offline system constraints. The proposed reformulation enables complexity-performance trade-offs which where not evident before. Finally, we also present a low complexity greedy algorithm which seems to approximate the original problem reasonably well. 


\section{Polar Codes}\label{sec:background}
Following the notation of \cite{Arikan2009}, we use $a_1^N$ to denote a row vector $(a_1,\hdots,a_N)$ and $a_i^j$ to denote the subvector $(a_i,\hdots,a_j)$. If $j < i$, then the subvector $a_i^j$ is empty. We use $\log(\cdot)$ to denote the binary logarithm.
\subsection{Construction of Polar Codes}
Let $W$ denote a binary input discrete and memoryless channel with input $u~\in~\{0,1\}$, output $y~\in~\mathcal{Y}$, and transition probabilities $W(y|u)$. A polar code is constructed by  applying a $2 \times 2$ \emph{channel combining} transformation recursively on $W$ for $n$ times, followed by a \emph{channel splitting} step~\cite{Arikan2009}. This results in a set of $N = 2^n$ channels, denoted by $W_N^{(i)}(y_1^N,u_1^{i-1}|u_i),~i=1,\hdots,N$. In principle, it is possible to compute the mutual information values $I(Y_1^N,U_1^{i-1}|U_i),~i=1,\hdots,N$. In practice, finding an analytical expression turns out to be a very hard problem, except for the case of the binary erasure channel (BEC), where an exact recursive calculation is possible \cite{Arikan2009}. Methods for approximating the mutual information values in more general cases are described in \cite{Arikan2009,Pedarsani2011,Tal2013}. The construction of a polar code is completed by choosing the good channels as \emph{non-frozen} channels which carry information bits, while \emph{freezing} the remaining channels to some known values $u_i$. The set of frozen channel indices is denoted by $\mathcal{A}^c$ and the set of non-frozen channel indices is denoted by $\mathcal{A}$.

\subsection{Successive Cancellation Decoding of Polar Codes}
In the SC decoding algorithm~\cite{Arikan2009}, decoding starts by computing an estimate of $u_1$, denoted by $\hat{u}_1$, based only on $y_1^N$. Subsequently, $u_2$ is estimated using  $(y_1^N,\hat{u}_1),$ etc. Decisions are taken according to
\begin{align}
	\hat{u}_i & =\left\{ \begin{matrix} \arg \max _{u_i \in \{0,1\}} W_N^{(i)}(y_1^N,\hat{u}_1^{i-1}|u_i), & i \in \mathcal{A}, \\ u_i, & i \in \mathcal{A}^c. \end{matrix} \right. \label{eq:scdec}
\end{align}
The channel likelihoods $W(y_i|x_i),~x_i\in \{0,1\}$ are combined through the stages of a decoding graph in order to calculate  $W_N^{(i)}(y_1^N,\hat{u}_1^{i-1}|u_i),~i=1,\hdots,N$~\cite{Arikan2009}. The decoding graph contains $N \log N$ nodes. If intermediate results are stored, then each node has to be activated only once during decoding. Thus, we need exactly $N \log N$ node-computations per codeword.

\subsection{Complexity Reduction Through Pruning}\label{sec:simplified}
Complexity reduction can be achieved by pruning nodes from the decoding graph whose descendant nodes at stage $0$ all correspond to frozen channels, since these nodes calculate likelihoods that will never actually be used by the decoding rule~\cite{Alamdar-Yazdi2011}. For example, consider a rate-1/2 code corresponding to the decoding graph of Fig.~\ref{fig:grapha} and let $u_1, u_3 \in \mathcal{A}^c$ and $u_2, u_4 \in \mathcal{A}$. In this case, (\ref{eq:scdec}) does not require the likelihoods $W_N^{(1)}(y_1^N|u_1),~u_1=0,1,$ and $W_N^{(3)}(y_1^N,\hat{u}_1^{2}|u_3),~u_3=0,1,$ to estimate $u_1$ and $u_3$, respectively. So, the corresponding node-computations at stage $0$ can be pruned. However, the computations at stage $1$ can not be pruned, since their results are required to estimate $u_2$ and $u_4$. If, instead, $u_1,u_2$ were chosen as frozen, then the computations for $u_1$ and $u_2$ at stage 0, as well as the two preceding computations at stage $1$ could be pruned. In the first example we can prune two node-computations, while in the second example we can prune four node-computations. However, in the second case the error rate performance of the code will be worse since we do not allow the two best channels to carry information.

\section{Enabling Complexity-Performance Trade-Offs}\label{sec:opt}



\subsection{Complexity and Performance Metrics}
The total number of computations that can be saved by pruning the decoding graph, denoted by $c$, is used as a complexity metric. Let the blocklength $N$ and the rate $R = 1 - \frac{k}{N},~k\in\mathbb{N},~0 < k < N$ be fixed and let the mutual information values of the $N$ channels be denoted by $I_i,~i=1\hdots,N$. We use the sum mutual information of the set of non-frozen channels as a performance metric, i.e.,
\begin{equation}
	m = \sum _{i \in \mathcal{A}} I_i = N \cdot I(W) - \sum _{i \in \mathcal{A}^c} I_i. \label{eqn:perfmetric}
\end{equation}
Note that the polar code construction proposed in~\cite{Arikan2009} maximizes this metric under the constraint $|\mathcal{A}| = k$ and let $m_{\max}$ denote this maximum, i.e., 
\begin{equation}
	m_{\max} = \max _{\mathcal{A}: |\mathcal{A}| = k} \sum _{i \in \mathcal{A}} I_i.
\end{equation}
Since $I_i \geq 0,~0 \leq i \leq N,$ the maximization amounts to selecting the channel indices with the $k$ largest $I_i$ values.

\subsection{Optimization Problem Formulation}
From a complexity perspective, it is favorable to form clusters of $2^l,~l\in\mathbb{N}$, frozen channels in order to maximize pruning. In this section, we describe an optimization problem which constructs a polar code of rate $R$, in a way that maximizes $c$ while ensuring that $m$ is larger than a predefined performance constraint $m' \geq 0$. To this end, the indices of the $N$ channels are grouped into clusters of $1,2,\hdots,N$ consecutive channels as illustrated in Fig.~\ref{fig:grapha}, where the illustration of the groups has been spread across the stages of the data dependency graph to reduce congestion. Let the set of all the groups be denoted by $\mathcal{G}$. We have
\begin{align}
	|\mathcal{G}| = N\sum _{j=0}^n2^{-j} = 2N-1. \label{eqn:varnum}
\end{align}

\begin{figure}
	\centering
	\includegraphics[width=0.3\textwidth]{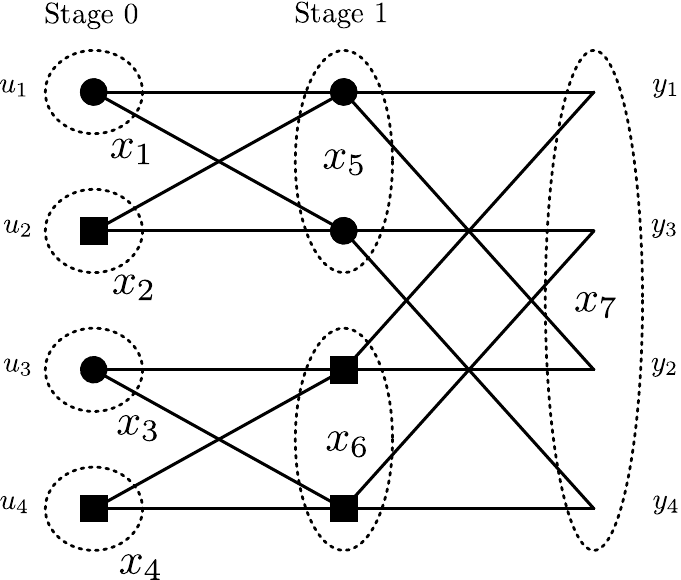}
	\caption{Decoding graph for $N=4$ with channel groups. An optimization variable $x_i$ is associated with each group $g_i$. Setting $x_i=1$ corresponds to freezing all channels in $g_i$.}
	\label{fig:grapha}
	\vspace{-0.3cm}
\end{figure}
\begin{figure}
	\centering
	\includegraphics[width=0.3\textwidth]{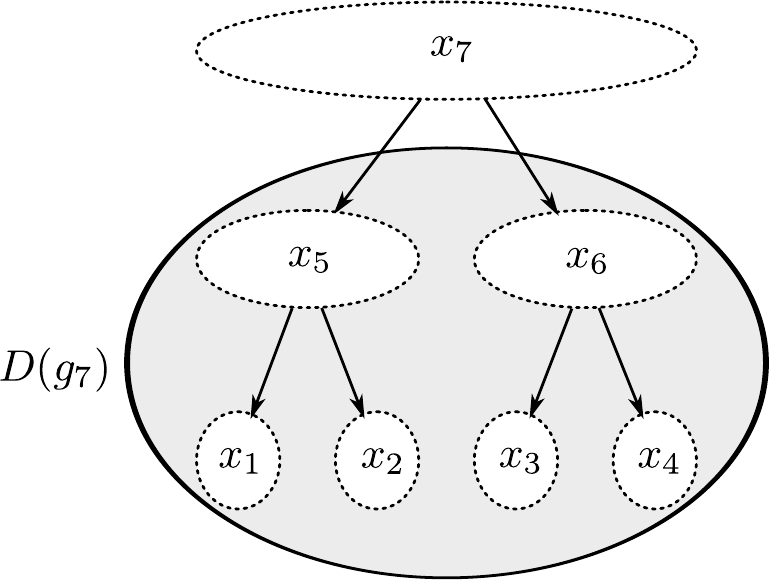}
	\caption{Tree structure of channel groups with descendants of $g_7$, i.e., $D(g_7)$, and their corresponding optimization variables. If $x_7=1$, then $x_i = 0$ has to be enforced for all $x_i: g_i \in D(g_7)$.}\label{fig:graphb}
	\vspace{-0.3cm}
\end{figure}
We associate each of the groups $g_i \in \mathcal{G}$ with a binary optimization variable $x_i,~i=1,\hdots,2N-1$. The assignment $x_i=1$ means that all channels contained in group $i$ are frozen. Each group also has a cost, denoted by $f_i,~i=1,\hdots,2N-1$. This cost is equal to the number of channel indices that are contained in $g_i$, i.e., $f_i = |g_i|$, and it reflects the rate loss incurred by setting $x_i = 1$. This leads to the \emph{rate constraint}
\begin{equation}
	\sum _{i=1}^{2N-1}f_ix_i = N - k. \label{eqn:size}
\end{equation}
Observe that, if in the example of Fig.~\ref{fig:grapha}, say, $x_7 = 1$, then the cost $f_7$ is paid. However, due to the tree structure of the groups, $f_7$ includes the costs for freezing the channels in groups $g_1$ to $g_6$, So, when $x_i = 1$ for any non-leaf group, $x_i = 0$ has to be enforced for all the descendants of this group in order not to count any costs more than once. Let the descendants of group $g_i \in \mathcal{G}$ be denoted by $D(g_i)$. An example is illustrated in Fig.~\ref{fig:graphb}. Let $\mathcal{M} = \{(i,j): g_i \in \mathcal{G}\backslash \{\text{leaves}\}, g_j\in D(g_i)\}$. Since $x_i \in \{0,1\}$, the \emph{mutual exclusiveness} constraint can be formalized as
\begin{align}
	x_i + x_j \leq 1,~\forall (i,j) \in \mathcal{M}. \label{eqn:mutex}
\end{align}
Moreover, we have
\begin{align}
	|\mathcal{M}| 	& = N\sum_{i=1}^{\log N-1} (\log N-i)2^{-i} = 2(\log N -1)N + 2. \label{eqn:mutexnum}
\end{align}
From (\ref{eqn:varnum}) and (\ref{eqn:mutexnum}), it can be seen that the number of variables grows linearly with the code length and the number of constraints in (\ref{eqn:mutex}) grows as $N \log N$. Each group $g_i \in \mathcal{G}$ has an associated gain in the number of computations, denoted by $c_i,~i=1,\hdots,2N-1$. This gain is the number of computations that is saved via pruning if all the channels in this group are frozen. Let $s(g_i) \in \{0,\hdots, \log N -1\}$ denote the stage to which group $g_i \in \mathcal{G}$ corresponds. For example, in Fig.~\ref{fig:grapha}, group $g_5$ corresponds to stage $1$. Then, we have
\begin{align}
	c_i	&  = (s(g_i)+1)2^{s(g_i)},\quad i=1,\hdots,2N-1. \label{eqn:complgain}
\end{align}
Due to (\ref{eqn:mutex}), no complexity gain is counted more than once. Finally, freezing the channels in group $g_i \in \mathcal{G}$ results in a loss in total mutual information, denoted by $m_i$, with 
\begin{equation}
	m_i = \sum_{j \in g_i} I_j,\quad i=1,\hdots,2N-1.
\end{equation}
again, due to (\ref{eqn:mutex}), no loss is counted more than once. A \emph{performance constraint} $m \geq m', ~m' \geq 0,$ is enforced, which can equivalently be written as
\begin{align}
	\sum _{i=1}^{2N-1} x_im_i  & \leq N \cdot I(W) - m'.
\end{align}
An optimization problem which maximizes the complexity gain, while ensuring that the resulting code has rate $R$ and satisfies the performance constraint, can be formulated~as
\begin{align}
	\text{maximize}	\quad	& \sum _{i=1}^{2N-1}c_ix_i \nonumber \\
	\text{subject to}\quad	& \sum _{i=1}^{2N-1}f_ix_i = N - k  \nonumber \\
							& \sum _{i=1}^{2N-1} x_im_i  \leq N \cdot I(W) - m' \label{prob:init} \\
							& x_i + x_j \leq 1,~\forall (i,j) \in \mathcal{M} \nonumber \\
							& x_i \in \{0,1\},\quad i = 1,\hdots, 2N-1 \nonumber
\end{align}
The above problem is a binary integer linear programming formulation of a multidimensional 0--1 knapsack problem~\cite{Freville2004}, which is known to be NP-hard. 
If $m'$ is chosen carefully so that $m' \leq m_{\max}$, then (\ref{prob:init}) is always feasible. Moreover, for $m' = m_{\max}$, the optimization problem reduces to the construction proposed by Ar{\i}kan,\footnote{Note that, in this case, the solution is not necessarily unique, but each solution of \eqref{prob:init} is also a solution of Ar{\i}kan's construction.} while $m' = 0$ results in a construction that maximizes the number of saved computations while completely disregarding performance. 
By varying $m'$ between these two extremal values, various complexity-performance trade-offs can be achieved.

\subsection{Results}

\begin{figure}
	\centering
	\includegraphics[width=0.4\textwidth]{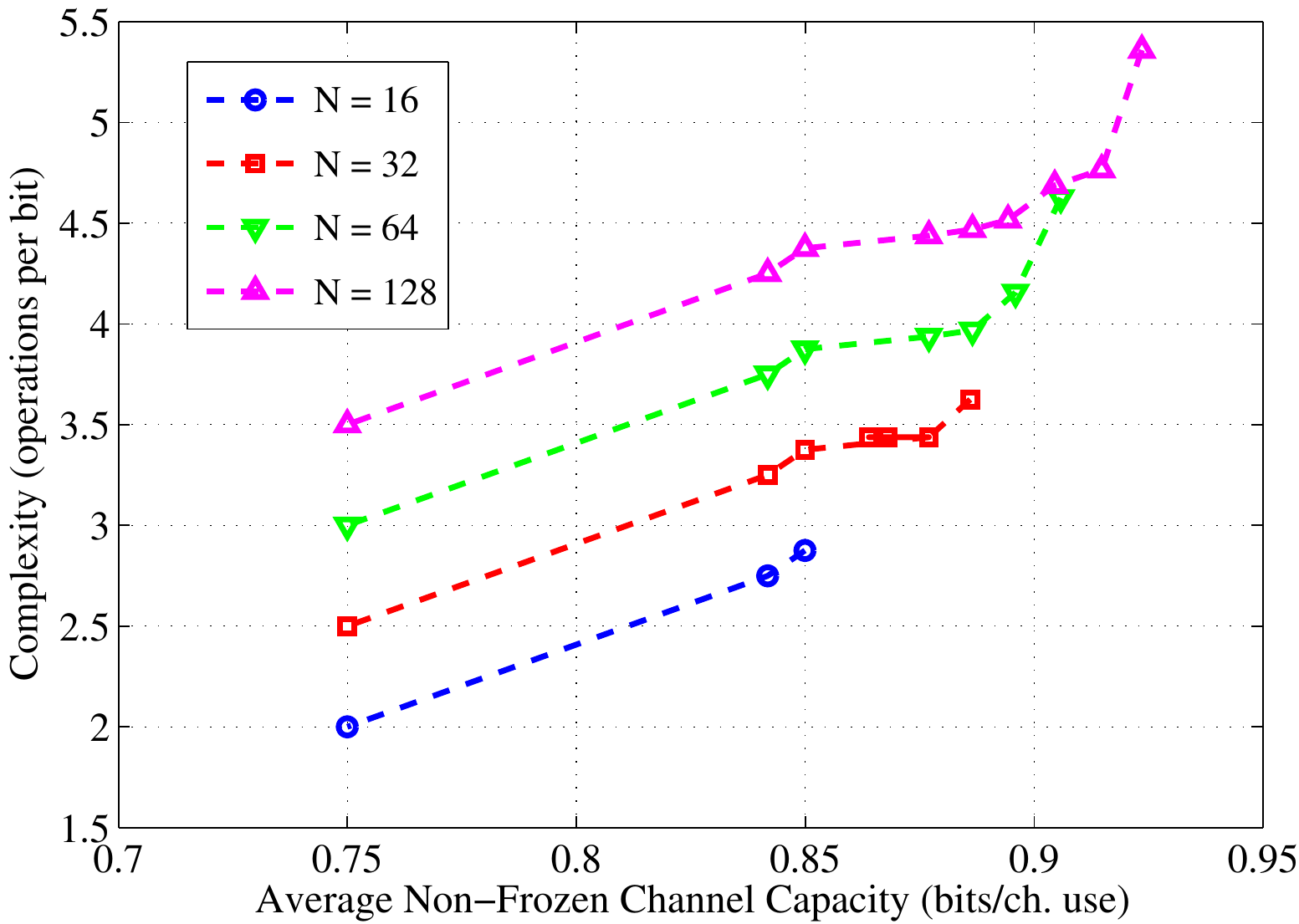}
	\caption{Solutions of \eqref{prob:init} for $R=0.5$, $N~=~2^n,~n=4,5,6,7,$ and transmission over a BEC$(0.5)$.}\label{fig:pareto}
	\vspace{-0.3cm}
\end{figure}

Even though \eqref{prob:init} is NP-hard, relatively small instances can still be solved by using standard branch-and-bound methods. For simplicity in calculating the mutual information values $I_i,~i=1,\hdots,N$, we present results only for the BEC$(\epsilon)$, where $\epsilon$ denotes the erasure probability. However, the proposed approach can be used for any other channel and input distribution, provided that $I_i,~i=1,\hdots,N$, are available. Moreover, given $I_i,~i=1,\hdots,N$, the complexity of \eqref{prob:init} and of the greedy algorithm presented in Section~\ref{sec:greedy} does not depend on the type of channel. We assume that the capacity achieving input distribution is used, so that $I(W) = 1 - \epsilon$. In Fig.~\ref{fig:pareto}, we present the solutions of (\ref{prob:init}) for transmission over a BEC$(0.5)$ with a polar code of rate $R = 0.5$ for $N=2^4,2^5,2^6,2^7$, which are obtained by solving the problem for various $0 \leq m' \leq m_{\max}$. We use the complexity in operations per bit on the vertical axis and the average mutual information on the horizontal axis. The former can be easily obtained from any solution $x^*$ as $\frac{1}{N}\left(N\log N - \sum _{i=1}^{2N-1}c_ix^*_i \right)$, while the latter is equal to $1 + \frac{1}{RN}\sum _{i=1}^{2N-1}m_ix^*_i$.


\section{Greedy Algorithm}\label{sec:greedy}
In order to solve \eqref{prob:init} for practically relevant blocklengths, like $2^{10} \leq N \leq 2^{20}$, in reasonable time, we present a greedy algorithm that takes advantage of the structure of the problem to provide useful solutions with negligible running time.

\subsection{Greedy Algorithm Description}
Our greedy algorithm consists of three steps, namely the \emph{greedy maximization} step, the \emph{feasibility} step, and the \emph{post-processing} step. In the first step, the goal is to greedily maximize the objective function while satisfying all inequality constraints. The second step ensures that the equality constraint is also satisfied, while the last step finalizes and improves the solution. Recall that $k' = N-k$ is the number of bits that need to be frozen. Let $k'_{\text{bin}}$ denote the $\log N$ bit right-MSB binary representation of $k'$ and let $k'_{\text{bin}}(j),~0 \leq j \leq \log N-1,$ denote the $j$-th bit of $k'_{\text{bin}}$. The greedy maximization step is inspired by the following observation.
\begin{prop}\label{prop:optimal}
If there were no performance constraint present in (\ref{prob:init}), the problem could be solved exactly as follows.
\begin{enumerate}
	\item[1.] Set $j = \log N-1$ and $x_i = 0,~1 \leq i \leq 2N-1$.
	\item[2.] If $k'_{\text{bin}}(j) = 1$, then set $x_i = 1$ for one $g_i: s(g_i) =  j$, denoted by $g_{i'}$, and set $x_i = 0$ for all remaining $g_i: s(g_i) =  j$. Remove all $x_i: g_i \in D(g_{i'})$ from the problem. 
	\item[3.] Set $j = j - 1$ and go to 2. until $j < 0$.
\end{enumerate}
\end{prop}
\begin{proof}
By eliminating all $x_i: g_i \in D(g_{i'})$ from the problem at step 2, we guarantee that the mutual exclusiveness constraint is not violated. Moreover, stage $\log N -1$ contains two groups, of which only one can be frozen, and for each group in stage $j$ there are two groups in stage $j-1$, so that step 2 can always be executed. We now show that any optimal solution must freeze at most one group per stage. Suppose that, for some solution, more than one groups were frozen in some stage $j$. Then, it is possible to replace any two frozen groups at stage $j$ with some frozen group at stage $j+1$ without violating any constraint. Based on \eqref{eqn:complgain}, for the complexity gains we have
\begin{equation}
	2 \cdot \left(j2^{j-1}\right) = j 2^{j} < (j+1) 2^{j},~\forall j \geq 0, \label{eqn:stagegain}
\end{equation} 
so this would strictly increase the objective function, meaning that the original solution could not have been optimal. Since all groups in stage $j$ contain $2^{(j-1)}$ bits and the binary representation of $k'$ is unique, it follows that the only way to freeze exactly $k'$ channels by freezing at most one group per stage, thus satisfying the rate constraint, is to freeze the groups according to the pattern dictated by $k'_{\text{bin}}$.
\end{proof}
\subsubsection{Greedy maximization step}
The greedy maximization step is different than the procedure of Proposition \ref{prop:optimal} in that it makes sure that the performance constraint is satisfied. In the following procedure, $k'_{\text{bin}}$ is again initialized to $\log N$ bit right-MSB binary representation of $k'$, but $k'_{\text{bin}}(j) \in \mathbb{N}$.
\begin{enumerate}
	\item[1.] Set $j = \log N-1$ and $x_i = 0,~1 \leq i \leq 2N-1$.
	\item[2.] If $k'_{\text{bin}}(j) \geq 1$, then try the following. 
		\begin{itemize}
			\item[2.1.] Find the $g_i: s(g_i) =  j$ with the smallest $m_i$ in stage $j$ and set $x_i = 1$. 
			\item[2.2.] If $\sum _{i} x_im_i  \leq N \cdot I(W) - m'$, then remove all $x_i: g_i \in D(g_{i'})$ from the problem, set $k'_{\text{bin}}(j) = k'_{\text{bin}}(j) - 1$, and go to 2.
			\item[2.3.] Else, set $k'_{\text{bin}}(j-1) = k'_{\text{bin}}(j-1) + 2$, set $x_i = 0$, and go to 3.
		\end{itemize}
	\item[3.] Set $j = j - 1$ and go to 2. until $j < 0$.
\end{enumerate}
At step 2.3., we set $k'_{\text{bin}}(j-1) = k'_{\text{bin}}(j-1) + 2$ because for each group that could not be frozen at stage $j$ due to the performance constraint, we need to freeze two groups at stage $(j-1)$ in order to (hopefully) satisfy the rate constraint. Unfortunately, there is no longer a guarantee that the procedure will be able to freeze exactly $k'$ bits as required to satisfy the rate constraint. However, the mutual exclusiveness and performance constraints are guaranteed to be met.

\begin{figure}
	\centering
	\includegraphics[width=0.42\textwidth]{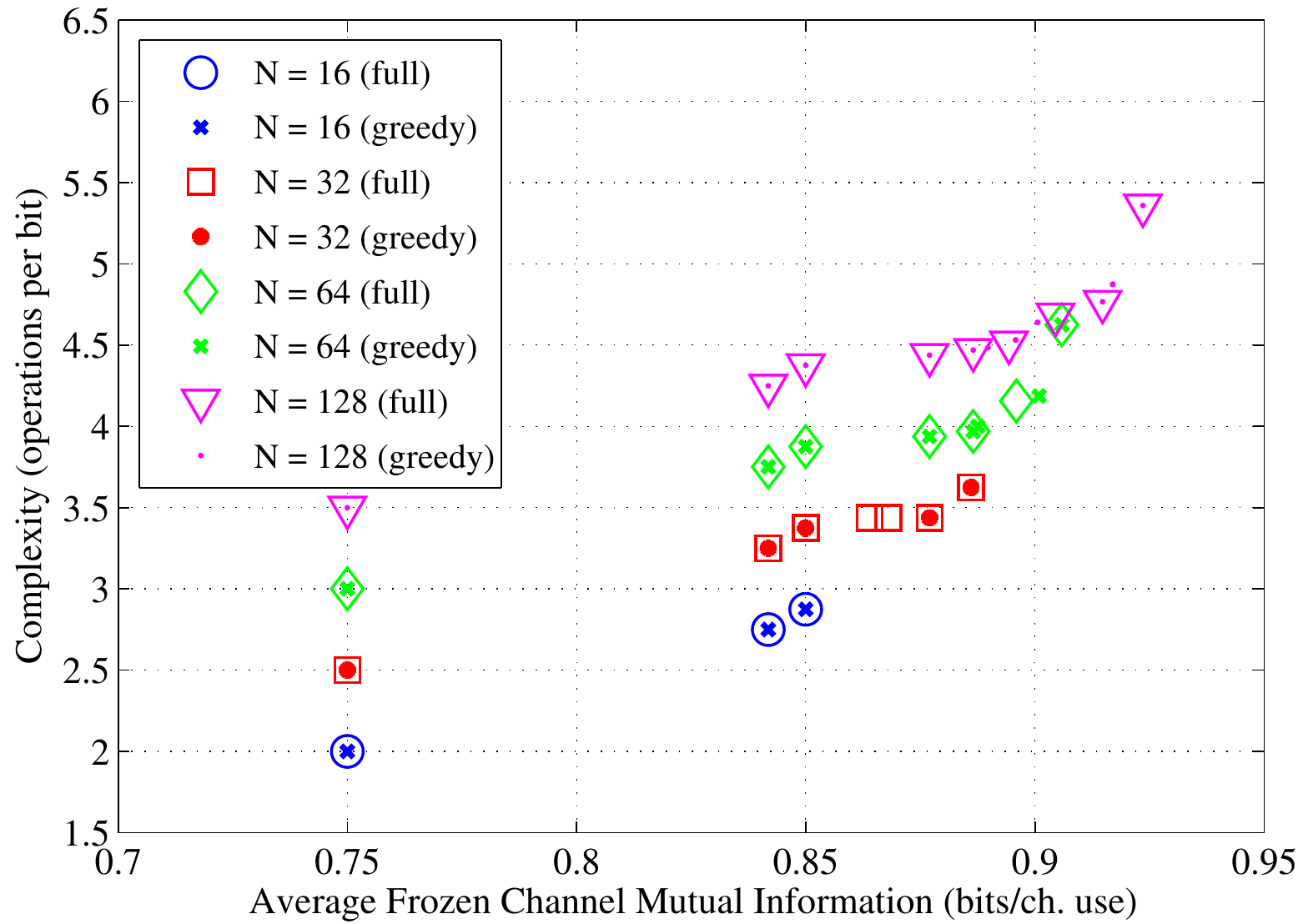}
	\caption{Results from exact solution of (\ref{prob:init}) and of the greedy algorithm for $R=0.5$, $N~=~2^n,~n=4,5,6,7,$ and transmission over a BEC$(0.5)$.}\label{fig:fullvsgreedy}
	\vspace{-0.25cm}
\end{figure}

\subsubsection{Feasibility step}
The second step of the algorithm sacrifices the objective function in a systematic step-by-step fashion until the solution is feasible, i.e., until the rate constraint is satisfied. Let $k''$ denote the number of additional bits that need to be frozen after the greedy maximization step is finished so that the rate constraint is satisfied, i.e., $k'' = k' - \sum _if_ix_i$. 

If $k'' > 0,$ then the feasibility step starts greedily unfreezing frozen groups to free up mutual information. More and more groups are unfrozen until the total number of unfrozen groups that can be frozen at stage $0$ is equal to $k''$ plus the number of variables in the groups that were unfrozen so far. Since during this step only groups at stage $0$ are refrozen which provide the smallest complexity gain, no direct effort is made to minimize the loss in the objective function. The feasibility step starts at stage $\lceil \log k'' \rceil + 1$, because by unfreezing a group in this stage it is possible to satisfy the rate constraint in a single step, thus making an indirect effort to minimize the objective function loss. Subsequently, all stages up to $\log N - 1$ are visited, and the procedure continues with stages $0$ to $\lceil \log k'' \rceil$, thus visiting all stages, if required. If $m' \leq m_{\max}$, the feasibility step is guaranteed to find a feasible solution.

\subsubsection{Post-processing step}
The post-processing step identifies pairs of consecutive frozen groups at each stage $j$ and replaces them with their parent group at stage $j+1$, which improves the objective function without violating any of the constraints.

\subsection{Results}
The solutions obtained by solving (\ref{prob:init}) exactly as well as by using the greedy algorithm for various constraints and blocklength up to $N = 2^7$ and for $R=0.50$ are compared in Fig.~\ref{fig:fullvsgreedy}. The greedy algorithm is able to find most of the optimal solutions for small instances of the problem. 

The solutions found by the greedy algorithm are presented in Fig.~\ref{fig:greedy} for various blocklengths and for $R=0.50$. For $N = 2^{20}$ the average running time of the greedy algorithm is less than $10^2$ seconds, which is negligible given that the optimization is carried out offline. We observe that the rightmost part of the curve is relatively steep, thus providing favorable trade-offs. For a fixed blocklength, the codes corresponding to some solution points can be chosen and stored in order to provide the system with online performance-complexity trade-offs. Moreover, during the design phase one can choose the solution with the largest performance among all blocklengths that satisfies a given complexity constraint.

\begin{figure}
	\centering
	\includegraphics[width=0.42\textwidth]{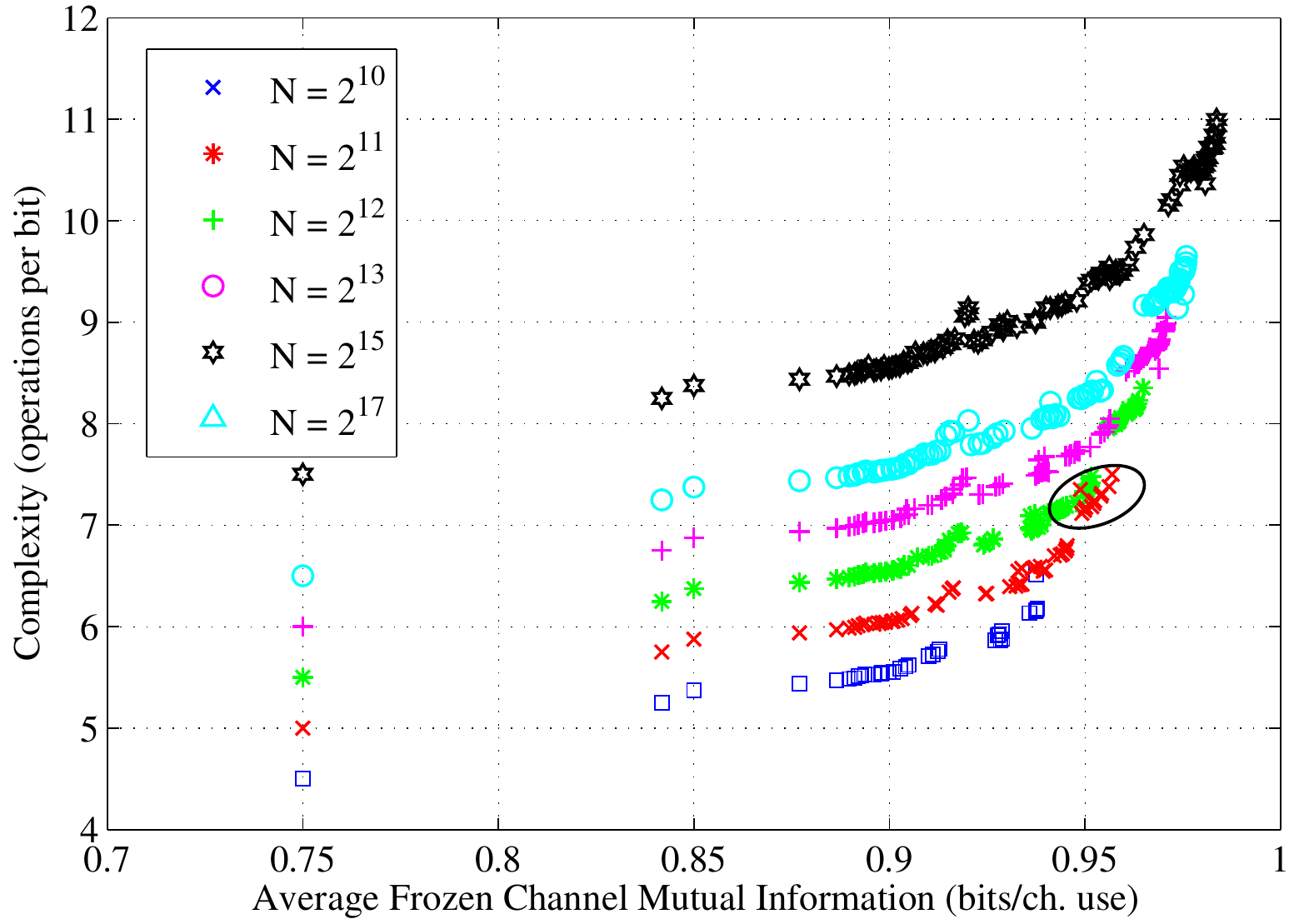}
	\caption{Solutions of greedy algorithm for $R=0.5$, $N~=~2^n,~n=9,10,11,12,13,15,$ over a BEC$(0.5)$. The circled codes for $N = 2^{10}$ are simulated in Fig.~\ref{fig:degradation}.}\label{fig:greedy}
	\vspace{-0.3cm}
\end{figure}

\section{Discussion}\label{sec:disc}

Our choice of performance metric requires some intuitive justification. Let $Z(W)$ denote the Bhattacharyya parameter of a channel $W$ and let $Z_i = Z(W^{(i)})$. It is known that $\sum _{i \in \mathcal{A}} Z_i$ is an upper bound on the probability of block error \cite{Arikan2009}. It was shown in \cite{Bastani2013} that, for the BEC, it is also a lower bound. Moreover, for the BEC we have $I_i = 1 - Z_i$, so by maximizing $\sum _{i \in \mathcal{A}}I_i$ one can minimize the probability of block error. Similarly, by placing a constraint on $\sum _{i \in \mathcal{A}}I_i$, we are implicitly placing a constraint on $\sum _{i \in \mathcal{A}}Z_i$, which is directly related with the probability of block error. So, for the case of the BEC, the metric that we use has an explicit relation with the probability of block error. For more general channels, one intuitively expects there to be at least an implicit relation between the two quantities. Ideally, one would like to use the probability of block error itself as a metric, but, to the best of our knowledge, this can not be described analytically as a function of $\mathcal{A}$, and especially not in a linear way which enables a simple formulation of the optimization problem.

Moreover, in principle, it is possible that a solution of \eqref{prob:init} contains a very bad channel in $\mathcal{A}$. This would lead to a catastrophic failure of the code, resulting in a block error rate (BLER) close to 1. This problem can be circumvented by adding the following additional constraints to \eqref{prob:init}
\begin{align}
	(1-x_i) \cdot h_i = 0,~i=1,\hdots,2N-1, \label{eqn:minqualconst}
\end{align}
where $h_i = 1$ if $g_i \in \mathcal{G}$ contains a channel with $I_i \leq m''$, where $m''$ is chosen as the lowest acceptable mutual information of the channels used for the information bits, and $h_i = 0$ otherwise. However, we have observed in simulations that the useful codes (a code is said to be \emph{useful} if it lies on the Pareto frontier of the set of obtained solutions) have performance which degrades gracefully with decreasing values of the performance metric. An example of this behavior for $N=2^{10}$ can be seen in Fig.~\ref{fig:degradation}, where code 1 corresponds to the construction in \cite{Arikan2009}, while codes 2 to 8 provide different performance-complexity trade-offs (annotated in Fig.~\ref{fig:greedy}).

\begin{figure}
	\centering
	\includegraphics[width=0.42\textwidth]{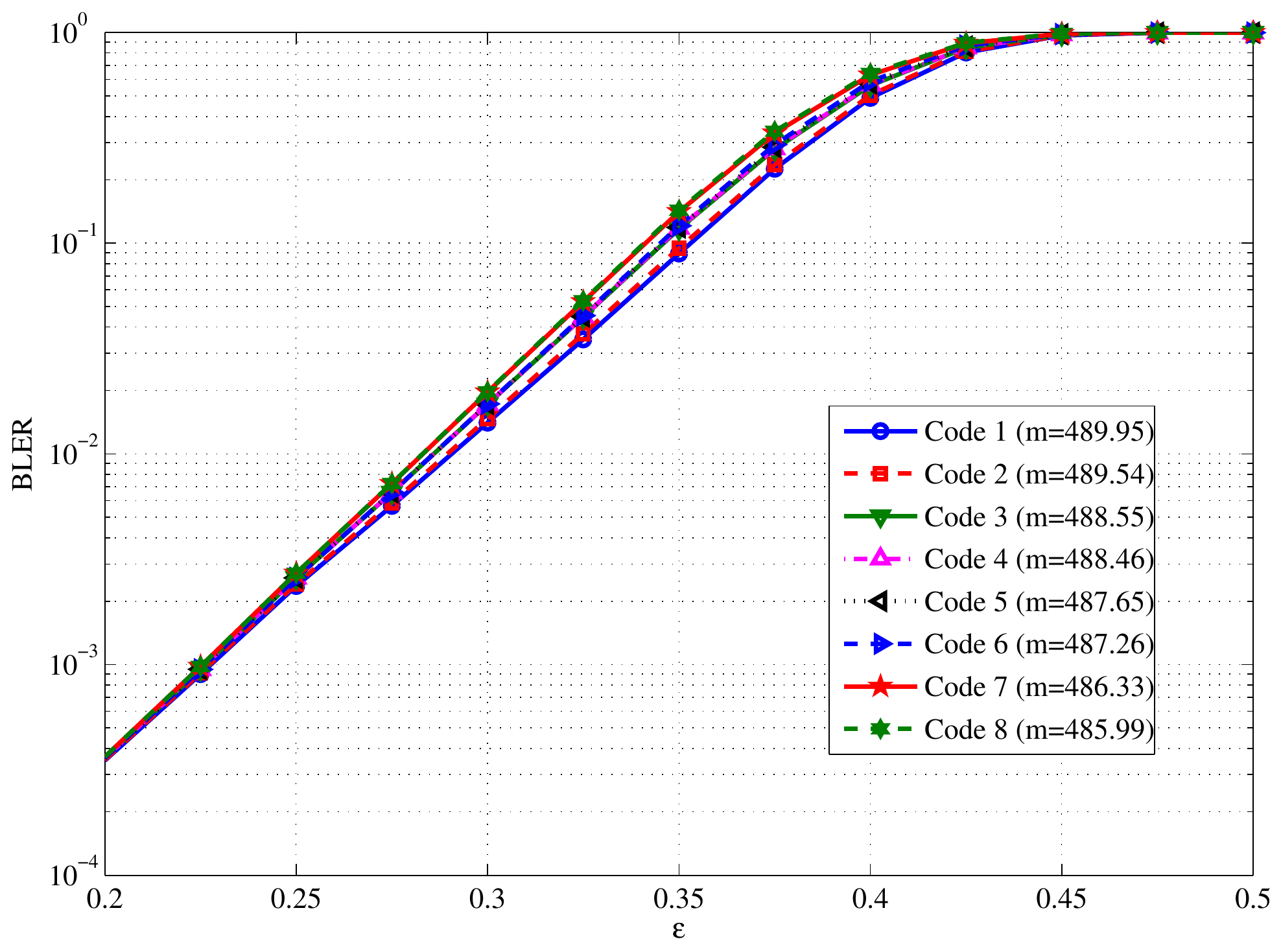}
	\caption{BLER performance and performance metric of the useful codes for $N = 2^{10}$.}\label{fig:degradation}
	\vspace{-0.3cm}
\end{figure}

\section{Conclusion}\label{sec:concl}
In this paper, we showed how to achieve fine-grained trade-offs between complexity and performance of SC decoding of polar codes by reformulating the frozen channel selection step of the standard polar code construction procedure as a 0-1 knapsack problem. Moreover, we described a low-complexity greedy algorithm, which is tailored to fit our specific knapsack problem instance. The greedy algorithm was used to approximately solve the optimization problem in order to construct polar codes of blocklength up to $N~=~2^{20}$.

\section*{Acknowledgment}
The authors would like to thank the anonymous reviewers for their helpful comments. This work was kindly supported by the Swiss NSF under Project ID 200021\_149447 and by the European Union under Marie Curie grant 304186.


\ifCLASSOPTIONcaptionsoff
  \newpage
\fi

\end{document}